\documentclass[10pt, conference]{IEEEtran} 

\usepackage{url}
\usepackage{amsfonts}
\usepackage{amssymb}
\usepackage{amsmath}
\usepackage{graphicx}
\usepackage{epsfig}
\usepackage{bm} 
\usepackage{bbm} 
\usepackage{enumerate}

\newcommand{\beq}[1]{\begin{equation}\label{#1}}
\newcommand{\eeq}{\end{equation}}

\newcommand{\beqn}[1]{\begin{eqnarray}\label{#1}}
\newcommand{\eeqn}{\end{eqnarray}}

\newtheorem{thmbody}{Theorem}
\newenvironment{thm}{
\begin{thmbody}
	}{
	\end{thmbody} 
	}
\newtheorem{dfnbody}{Definition}

\newtheorem{corbody}{Corollary}
\newenvironment{cor}{
\begin{corbody}
	}{
	\end{corbody} 
	}
\newtheorem{lemmabody}{Lemma}
\newenvironment{lemma}{
\begin{lemmabody}
	}{
	\end{lemmabody} 
	}
\newtheorem{propbody}{Proposition}

\newenvironment{proof}{
	{\it Proof:}
	}{
 $\Box$
	}

\usepackage{dsfont}
\begin{document}
\title{Exponential Source/Channel Duality}


\author{\IEEEauthorblockN{Sergey Tridenski}
\IEEEauthorblockA{EE - Systems Department\\Tel Aviv University, Israel\\
Email: sergeytr@post.tau.ac.il}
\and
\IEEEauthorblockN{Ram Zamir}
\IEEEauthorblockA{EE - Systems Department\\Tel Aviv University, Israel\\
Email: zamir@eng.tau.ac.il}}



\maketitle
\begin{abstract}
%
We propose a source/channel duality in the exponential regime,
where success/failure in source coding parallels
error/correctness in channel coding,
and a distortion constraint becomes
a log-likelihood ratio (LLR) threshold.
We establish this duality by first deriving exact exponents
for lossy coding of a memoryless source $P$,
at distortion $D$, for a general i.i.d. codebook distribution $Q$,
for both encoding success ($R<R(P,Q,D)$) and failure ($R>R(P,Q,D)$).
We then turn to maximum likelihood (ML) decoding over
a memoryless channel $P$ with an i.i.d. input $Q$,
and show that if we substitute $P=QP$, $Q=Q$,
and $D=0$ under the LLR distortion measure,
then the exact exponents for decoding-error
($R < I(Q, P)$) and strict correct-decoding ($R > I(Q, P)$)
follow as special cases of the exponents for source encoding success/failure, respectively.
Moreover, by letting the threshold $D$ take general values,
the exact random-coding exponents for erasure ($D>0$) and
list decoding ($D<0$) under the simplified Forney decoder
are obtained.
Finally, we derive the exact random-coding exponent for Forney's optimum tradeoff erasure/list decoder, and show that at the erasure regime it coincides with Forney's lower bound and with the simplified decoder exponent.
\footnote{This research was supported in part by the Israel Academy of Science, grant \# 676/15.}
\end{abstract}



%
%
%
%


\section{Introduction} \label{Introduction}
The aim of our study is to define an analogy
between lossy source coding and coding through a noisy channel,
allowing to translate the terms and results between the two branches.
We consider an analogy,
in which encoding for sources corresponds to
decoding for channels, 
{\em encoding success} translates to {\em decoding error},
and 
a source 
translates to 
a channel input-output pair.
Channel coding, in the random coding setting with a fixed generating distribution $Q$, emerges as a special case of lossy source coding.
Although other analogies may be possible,
the proposed analogy 
requires minimal generalization on
the source coding part.
Generalization of the channel decoder, on the other hand, leads to a broader correspondence
between the two branches, 
such that {\em correct-decoding} event for channels
(which becomes a rare event for a sufficiently large codebook)
translates to {\em encoding failure}
for sources.

In other works on the source/channel duality,
the rate-distortion function of a DMS $P(x)$
\begin{displaymath}
R(P, D) \;\; = \;\; \min_{\\W(\hat{x} \,|\, x): \;\; \mathbb{E}\,\left[d(X, \,\hat{X})\right] \; \leq \; D} \; I(X; \hat{X}),
\end{displaymath}
and the capacity of a DMC $P(y \, | \, x)$
\begin{displaymath}
C(P) \;\; = \;\; \max_{Q(x)} \; I(X; Y)
\end{displaymath}
are related directly via introduction of a cost constraint on the channel input
\cite{GuptaVerdu11}, \cite{PradhanChouRamchandran03}, \cite{Blahut74},
or using covering/packing duality \cite{JanaBlahutISIT2006}.
In order to look at the similarities and the differences of
the expressions for the RDF and capacity more closely,
let us rewrite them in a unified fashion,
with the help of the Kullback-Leibler divergence $D(\cdot \| \cdot)$,
as follows:
\begin{equation} \label{eqRPQDdetailed}
R(P, D) \; = \; \min_{Q(\hat{x})}
\min_{\substack{\\W(\hat{x} \,|\, x): \\
\\ d(P\,\circ\, W) \; \leq \; D}} D(P \circ W \, \| \, P \times Q),
\end{equation}
\begin{equation} \label{eqRQPQ0detailed}
C(P) \; = \; \max_{Q(\cdot)}
\min_{\substack{\\W(\hat{x} \,|\, x, \, y): \\
\\ d((Q\,\circ \,P)\,\circ\, W) \; \leq \; 0}} D\big((Q\circ P) \circ W \, \| \, (Q\circ P) \times Q\big),
\end{equation}
where in the case of the capacity we use a {\em particular} distortion measure,
defined by the LLR as
\begin{equation} \label{eqDmeasure}
d\big((x, y), \hat{x}\big) \; \triangleq \; \ln \frac{P(y \,|\, x)}{P(y \,|\, \hat{x})}.
\end{equation}
Note, that the distortion constraint in the capacity case (\ref{eqRQPQ0detailed}) is $D = 0$.
More concisely, the expressions (\ref{eqRPQDdetailed}) and (\ref{eqRQPQ0detailed}) may be written with the help of the function
$R(P, Q, D)$, 
defined as the inner $\min$ in (\ref{eqRPQDdetailed}),
\cite{ZamirRose01},
which represents the rate in lossy coding of a source $P$
using an i.i.d. codebook $Q$ under a distortion constraint $D$:
\begin{align}
R(P, D) \; & = \; \min_{\substack{\\Q}}\; R(P, Q, D),
\label{eqRPQD} \\
C(P) \; & = \; \max_{\substack{\\Q}} \; R(Q\circ P, \,Q, \,0).
\label{eqRQPQ0}
\end{align}
The expression for the capacity (\ref{eqRQPQ0detailed}), or (\ref{eqRQPQ0}), follows by our new results, and it can also be shown directly,
by the method of Lagrange multipliers,
that, for the distortion measure (\ref{eqDmeasure}), $R(Q\circ P, \, Q, \, 0) \, = \, I(Q\circ P)$, where $I(Q\circ P)$ is the mutual information.
Obviously, $\min_{\,Q} R(P, Q, D)$ and $\max_{\,Q} R(Q\circ P, \,Q, \,0)$ are two mathematically different problems
and it is difficult to relate between them. On the other hand,
for a given $Q$, i.e. before minimization/maximization over $Q$, the expression for channels in (\ref{eqRQPQ0detailed}) is just a special case of the expression for sources in (\ref{eqRPQDdetailed}) or (\ref{eqRPQD}).
Therefore, in this work we focus on the source/channel analogy for a given $Q$.
In the rest of the paper, $Q$ plays the role of a generating distribution
of an i.i.d. random codebook.

We extend the described analogy to the framework of random coding exponents.
In our terminology, the source encoding failure is the same as the source coding error defined in \cite{Marton74}.
We derive asymptotically-exact exponents of source encoding failure and success,
which are similar in form to the best lossy source coding exponents given in \cite{Marton74} and \cite[p.~158]{CsiszarKorner},
respectively, but correspond to i.i.d. random code ensembles
generated according to an arbitrary (not necessarily optimal) $Q$.
Such ensembles prove to be useful for adaptive
source and channel coding \cite{ZamirRose01},
\cite{TridenskiZamirITW2015}.

Next, we modify the ML channel decoder with a threshold $D$.
The resulting decoder can be identified as
a simplified erasure/list decoder \cite[eq.~11a]{Forney68},
suggested by Forney as an approximation
to his optimum tradeoff decoder \cite[eq.~11]{Forney68}.
Exact random coding exponents for the simplified decoder,
via the source/channel analogy,
then become corollaries of our source coding exponents,
where Gallager's random coding error exponent of the ML
decoder is obtained as a special case for $D = 0$.

The fixed composition version of the random coding error exponent
for the simplified decoder
was derived recently in \cite{WeinbergerMerhavISIT2015}.
In comparison, the i.i.d. random coding error exponent, derived here,
can be expressed similarly to Forney's random coding bound for the optimum tradeoff decoder \cite[eq.~24]{Forney68},
and therefore can be easily compared to it.
We show that the exact i.i.d. random coding exponent of the simplified decoder
coincides with Forney's lower bound on the random coding exponent
\cite[eq.~24]{Forney68} for $T \equiv D \geq 0$.
It follows, that Forney's lower bound is tight for random codes, also with respect to the optimum tradeoff decoder, for $T \geq 0$.

Finally, we derive an exact random coding exponent
for Forney's optimum tradeoff decoder \cite[eq.~11]{Forney68}
for all values of the threshold $T$.
The resulting expression is also similar to the original Forney's random coding bound \cite[eq.~24]{Forney68}
and we show directly that they coincide for the threshold $T \geq 0$.
This proves a conjecture in \cite{BaruchMerhav11},
and also extends/improves the results in \cite{Merhav08}, \cite{BaruchMerhav11} for list decoding.

In what follows, we present our results for sources, then translate them to the results for channels.
We discuss briefly the relation of the new channel exponents to the error/correct exponents of the ML decoder.
Finally, we present our result for Forney's optimum tradeoff decoder.
In the remainder of the paper, a selected proof is given.
The rest of the proofs can be found in \cite{TridenskiZamir17}.


\section{Exponents for sources} \label{Sources}
We assume that the source alphabet ${\cal X} = \{x: P(x) > 0\}$
and the reproduction alphabet $\hat{\cal X} = \{\hat{x}: Q(\hat{x}) > 0\}$
are finite.
Assume also an additive distortion measure, of arbitrary sign, $d: {\cal X} \times \hat{\cal X} \rightarrow \mathbb{R}$,
such that the distortion
in a source-reproduction pair $({\bf x}, \hat{\bf x})$ of length $n$,
is given by $d({\bf x}, \hat{\bf x}) = \sum_{i\, = \, 1}^{n}d(x_{i}, \hat{x}_{i})$.
For an arbitrary distortion constraint $D$ and a distribution $T(x)$ over ${\cal X}$, let us define a function
\begin{equation} \label{eqRTQDDefinition}
R(T, Q, D) \; \triangleq \; \min_{W(\hat{x} \,|\, x): \;\; d(T\,\circ\, W) \; \leq \; D} D(T \circ W \, \| \, T \times Q),
\end{equation}
where $d(T\circ W) \triangleq \sum_{x, \, \hat{x}}T(x)W(\hat{x} \,|\, x)d(x, \hat{x})$.
If the set $\{W(\hat{x} \,|\, x): \; d(T\circ W) \, \leq \, D\}$
is empty, then $R(T, Q, D) \triangleq +\infty$.
For brevity, we define also the following function
\begin{equation} \label{eqE0}
{E\mathstrut}_{0}(s, \rho, Q, D)
\; \triangleq \;  -\ln \sum_{x}P(x)\Bigg[\sum_{\hat{x}}Q(\hat{x})e^{-s[d(x,\, \hat{x})-D]}\Bigg]^{\rho}.
\end{equation}

Consider a reproduction codebook of $M = e^{nR}$ random codewords
$\hat{\bf X}_{m}$ of length $n$,
generated i.i.d. according to the distribution $Q$.
Let {\bf X} be a random source sequence of length $n$ from the DMS $P$.
Let us define encoding success as an event
\begin{equation} \label{eqEncSuccEvent}
{\cal E}_{s} \;\; \triangleq \;\; \big\{\exists \, m\,: \;\;\; d({\bf X}, \hat{\bf X}_{m}) \;\; \leq \;\; nD\big\}.
\end{equation}
Then, our results for encoding success exponent can be formulated as follows:

{\em Definition 1 (Implicit formula):}
\begin{equation} \label{eqEsDefinition}
{E}_{s}(Q, R, D)
\; \triangleq \;
\min_{T(x)} \; \Big\{ D(T \| P) + {\big| R(T, Q, D) - R \big|\mathstrut}_{}^{+} \Big\}.
\end{equation}
Note that this expression is zero for $R \geq R(P, Q, D)$.
\begin{thm}[Explicit formula]\label{thm1} \hfill
\begin{equation} \label{eqES}
{E}_{s}(Q, R, D)
\; = \; \sup_{0 \, \leq \,\rho \,\leq \,1}
\;\;
\sup_{s\,\geq\,0} \;
\;
\big\{ {E\mathstrut}_{0}(s, \rho, Q, D) - \rho R \big\}.
\end{equation}
\end{thm}

\begin{thm}[Encoding success exponent]\label{thm2} \hfill
\begin{equation} \label{eqSuccess}
\lim_{n \, \rightarrow \, \infty} \; \left\{-\frac{1}{n}\ln \Pr\{{\cal E}_{s}\}\right\} \;\; = \;\;
{E}_{s}(Q, R, D),
\end{equation}
{\em except possibly for $D =  D_{\min} = \min_{x, \, \hat{x}}d(x, \hat{x})$, when the right-hand side is a lower bound.}
\end{thm}

Let us define encoding failure as a complementary event ${\cal E}_{f} \triangleq {{\cal E}_{s}\mathstrut}^{c}$.
Then, our results for encoding failure exponent can be formulated as follows:

{\em Definition 2 (Implicit formula):}
\begin{equation} \label{eqEfDefinition}
{E}_{f}(Q, R, D)
\; \triangleq \;
\min_{T(x): \;\; R(T, \,Q, \,D) \; \geq \; R} \; D(T \| P).
\end{equation}
Note that this expression is zero for $R \leq R(P, Q, D)$
and is considered $+\infty$ if $R > R_{\text{max}}(D) = \max_{\,T(x)}R(T, Q, D)$.

We give an explicit formula,
which does not always coincide with the implicit formula (\ref{eqEfDefinition}) for all $R$,
but gives the best convex ($\cup$) lower bound for (\ref{eqEfDefinition}) as a function of $R$, for sufficiently lax distortion constraint $D$:
\begin{thm}[Explicit formula]\label{thm3}
{\em For distortion constraint $\;D \; \geq \; \max_{x} \min_{\hat{x}} d(x, \hat{x})$,}
\begin{equation} \label{eqEF}
\text{\em l. c. e.}\;\big\{{E}_{f}(R)\big\} =
\sup_{\rho \,\geq \,0}
\;
\inf_{\substack{\\ s\,\geq\,0}} \;
\big\{ {E\mathstrut}_{0}(s, -\rho, Q, D) + \rho R \big\}.
\!\!
\end{equation}
{\em For $\;D \; < \; \max_{x} \min_{\hat{x}} d(x, \hat{x})$, the right-hand side expression gives zero, which is strictly lower than ${E}_{f}(Q, R, D)$,
if $R \, > \, R(P, Q, D)$.}
\end{thm}
Note that the above explicit formula is similar to the lower bound on the failure exponent given in
\cite{Blahut74} (except here it is without maximization over $Q$ and pertains to the random code ensemble of distribution $Q$).
Our result is also more specific about the relationship (convex envelope) between the lower bound and the true exponent, given by (\ref{eqEfDefinition})
according to
\begin{thm}[Encoding failure exponent]\label{thm4} \hfill
\begin{equation} \label{eqFailure}
\lim_{n \, \rightarrow \, \infty} \; \left\{-\frac{1}{n}\ln \Pr\{{\cal E}_{f}\}\right\} \;\; = \;\;
{E}_{f}(Q, R, D),
\end{equation}
{\em with the possible exception of points of discontinuity of the function ${E}_{f}(R, D)$.}
\end{thm}

\section{Exponents for channels} \label{Channels}
We assume a DMC with finite input and output alphabets
${\cal X}$ and ${\cal Y}$.
For simplicity,
we assume also that for any $(x, y)\in {\cal X} \times {\cal Y}$ the channel probability $P(y \, | \, x) > 0$.
Consider a codebook of $M = e^{nR} + 1$ random codewords ${\bf X}_{m}$ of length $n$,
generated i.i.d. according to a distribution $Q$ over ${\cal X}$.
Without loss of generality,
assume that message $m$ is transmitted.
Let {\bf Y} be a response, of length $n$, of the DMC $P$
to the input ${\bf X}_{m}$.
Let us define decoding error as an event:
\begin{equation} \label{eqErrorEvent}
{\cal E}_{e} \; \triangleq \;\left\{\exists \, m'\,\neq \,m\,: \;\;\; \ln\frac{P({\bf Y} \,|\, {\bf X}_{m})}{P({\bf Y} \,|\, {\bf X}_{m'})} \;\; \leq \;\; nD\right\},
\end{equation}
corresponding to
a simplified erasure/list decoder \cite[eq.~11a]{Forney68}.
Observe, from comparison of the events (\ref{eqEncSuccEvent}) and (\ref{eqErrorEvent}),
that the latter 
can be seen as
a special case of the former.
In (\ref{eqErrorEvent}), the channel input-output pair $({\bf X}_{m}, {\bf Y})$ pays a role analogous to the source sequence ${\bf X}$ in (\ref{eqEncSuccEvent}),
and the incorrect codeword ${\bf X}_{m'}$ plays a role analogous to the reproduction sequence $\hat{\bf X}_{m}$ in (\ref{eqEncSuccEvent}).
In the proposed analogy, the reproduction alphabet is the alphabet of incorrect codewords, which is ${\cal X}$,
and the alphabet of the source is the product alphabet of the channel input-output pair ${\cal X}\times {\cal Y}$.
We make the following substitutions:
\begin{align}
\hat{\cal X} \, = \, {\cal X} \;\;\;\;\;\;\;\;\; & \longrightarrow \;\;\;\;\;\;\;\;\; \hat{\cal X}
\label{eqSubstitutions0} \\
{\cal X}\times {\cal Y} \;\;\;\;\;\;\;\;\; & \longrightarrow \;\;\;\;\;\;\;\;\; {\cal X}
\label{eqSubstitutions1} \\
Q(x)P(y \,|\, x) \;\;\;\;\;\;\;\;\; & \longrightarrow \;\;\;\;\;\;\;\;\; P(x)
\label{eqSubstitutions2} \\
d\big((x, y), \hat{x}\big) \; = \; \ln \frac{P(y \,|\, x)}{P(y \,|\, \hat{x})}
\;\;\;\;\;\;\;\;\; & \longrightarrow \;\;\;\;\;\;\;\;\;
d(x, \hat{x})
\end{align}
Definition (\ref{eqE0}) now acquires a specific form
\begin{align}
&
{E\mathstrut}_{0}(s, \rho, Q, D)
\;\; \triangleq
\nonumber \\
& \!\!
-\ln\sum_{x, \,y}Q(x)P(y \,|\, x)\left[\sum_{\hat{x}}Q(\hat{x})\left[\frac{P(y \,|\, x)}{P(y \,|\, \hat{x})}\,e^{-D}\right]^{-s}\right]^{\rho}
\!\!\! .
\label{eqE0Chan}
\end{align}
Note, that the minimal distortion now depends on the support of the distribution $Q$:
\begin{equation} \label{eqDminQsupport}
D_{\min}(Q) \; \triangleq \; \min_{y}\;\min_{x, \, \hat{x}:\; Q(x)\cdot Q(\hat{x}) \, > \, 0}\;\ln \frac{P(y \,|\, x)}{P(y \,|\, \hat{x})}.
\end{equation}
The results for decoding error follow as simple corollaries of the results for encoding success.
The definition of the implicit expression for decoding error exponent parallels Definition~1:
\begin{equation} \label{eqEeDefinition}
{E}_{e}(Q, R, D)
\, \triangleq \,
\min_{T(x, \, y)} \Big\{ D(T \| Q \circ P) + {\big| R(T, Q, D) - R \big|\mathstrut}_{}^{+} \Big\},
\end{equation}
where $R(T, Q, D)$ is defined with $W(\hat{x} \,|\, x, \, y)$, as in (\ref{eqRQPQ0detailed}).
Note, that ${E}_{e}(Q, R, D)$ is zero for $R \geq R(Q\circ P, \,Q, \,D)$.
\begin{cor}[Explicit formula]\label{cor1} \hfill
\begin{equation} \label{eqExplErr}
{E}_{e}(Q, R, D)
= \sup_{0 \, \leq \,\rho \,\leq \,1}
\;
\sup_{s\,\geq\,0} \;
\;
\big\{ {E\mathstrut}_{0}(s, \rho, Q, D) - \rho R \big\}.
\end{equation}
\end{cor}
\begin{cor}[Decoding error exponent]\label{cor2} \hfill
\begin{equation} \label{eqError}
\lim_{n \, \rightarrow \, \infty} \; \left\{-\frac{1}{n}\ln \Pr\{{\cal E}_{e}\}\right\} \;\; = \;\;
{E}_{e}(Q, R, D),
\end{equation}
{\em except possibly for $D =  D_{\min}(Q)$, given by (\ref{eqDminQsupport}), when the right-hand side is a lower bound.}
\end{cor}
The best random coding exponent is given by
\begin{thm}[Maximal decoding error exponent]\label{thm5} \hfill
\begin{equation} \label{eqMaxErr}
\sup_{Q(x)}\; \lim_{n \, \rightarrow \, \infty} \; \left\{-\frac{1}{n}\ln \Pr\{{\cal E}_{e}\}\right\} \;\; = \;\;
\sup_{Q(x)}\; {E}_{e}(Q, R, D),
\end{equation}
{\em for all $(R, D)$.}
\end{thm}
This result can be contrasted with the fixed composition exponent \cite[eq.~29]{WeinbergerMerhavISIT2015},
and (together with the explicit form (\ref{eqExplErr})) can be easily compared with
Forney's random coding bound \cite[eq.~24]{Forney68}.

Similarly, the results for the correct decoding event ${\cal E}_{c} \triangleq {{\cal E}_{e}\mathstrut}^{c}$ follow as simple corollaries of the results for encoding failure.
The definition of the implicit expression for correct decoding exponent parallels Definition~2:
\begin{equation} \label{eqCorrectExtended}
{E}_{c}^{*}(Q, R, D) \; \triangleq \; \min_{T(x, \, y): \;\; R(T, \, Q, \, D) \; \geq \; R} \; D(T \; \| \; Q \circ P).
\end{equation}
The superscript $^{*}$ serves to indicate that this exponent is different from the correct decoding exponent of the ML decoder, for $D=0$,
as here the receiver declares an error also when there is only equality in (\ref{eqErrorEvent}), i.e. no tie-breaking. This distinction is important in the case of the correct-decoding exponent, but not in the case of the decoding error exponent.

The following explicit formula gives the best convex ($\cup$) lower bound for (\ref{eqCorrectExtended}) as a function of $R$, for nonnegative distortion constraint $D$:
\begin{cor}[Explicit formula]\label{cor3}
{\em For distortion constraint $\;D \; \geq \; 0$,}
\begin{equation} \label{eqCD}
\text{\em l. c. e.}\;\big\{{E}_{c}^{*}(R)\big\} =
\sup_{\rho \,\geq \,0}
\;
\inf_{\substack{\\ s\,\geq\,0}} \;
\big\{ {E\mathstrut}_{0}(s, -\rho, Q, D) + \rho R \big\}.
\!\!
\end{equation}
{\em For $\;D \; < \; 0$, the right-hand side expression gives zero, which is strictly lower than ${E}_{c}^{*}(Q, R, D)$,
if $R \, > \, R(Q\circ P, \, Q, \,D)$.}
\end{cor}

\begin{cor}[Correct decoding exponent]\label{cor4} \hfill
\begin{equation} \label{eqCorrect}
\lim_{n \, \rightarrow \, \infty} \; \left\{-\frac{1}{n}\ln \Pr\{{\cal E}_{c}\}\right\} \;\; = \;\;
{E}_{c}^{*}(Q, R, D),
\end{equation}
{\em with the possible exception of points of discontinuity of the function ${E}_{c}^{*}(R, D)$.}
\end{cor}

\section{Relation to the exponents of the ML decoder} \label{Regular}
The maximum likelihood decoder has the same error exponent as the decoder (\ref{eqErrorEvent}) with $D=0$.
The Gallager exponent \cite{Gallager73} is obtained from the explicit formula (\ref{eqExplErr}) with $D=0$.

On the other hand, the {\em correct-decoding} exponent of the ML decoder is given by
\begin{align}
&
E_{c}(Q, R) \; = \; \min_{T(x, \, y)} \; \Big\{ D(T \; \| \; Q \circ P) + {\big| R - R(T, Q, 0)\big|\mathstrut}_{}^{+} \Big\}
\label{eqCorrectD0} \\
&
\overset{(*)}{=}
\;
\sup_{0 \, \leq \,\rho \, < \,1}
\;
\big\{ {E\mathstrut}_{0}(\tfrac{1}{1 \, - \, \rho}, \,-\rho, \,Q, \,0) + \rho R \big\}
\label{eqCorrectRho} \\
& =
\min_{U(x), \, W(y\,|\,x)} \Big\{ D(U\circ W \,\|\, Q \circ P)
\nonumber \\
& \;\;\;\;\;\;\;\;\;\;\;\;\;\;\;\;\;\;\;\;\;\;\;\;
+ {\big| R - D(U\,\|\,Q) - I(U\circ W)\big|\mathstrut}_{}^{+} \Big\}
\label{eqCorrectImplicitAnother} \\
& \leq
\;\;\;\;
\min_{W(y\,|\,x)}
\;\;\;\, \Big\{ D(Q\circ W \,\|\, Q \circ P)
+ {\big| R - I(Q\circ W) \big|\mathstrut}_{}^{+} \Big\},
\label{eqDueckKorner}
\end{align}
where the equality ($*$) is shown in \cite{TridenskiZamir17},
(\ref{eqCorrectImplicitAnother}) is another implicit formula,
which is more convenient to derive (as well as convenient for the derivation of (\ref{eqCorrectRho})),
and (\ref{eqDueckKorner}) is the fixed composition exponent of Dueck and K{\"o}rner
\cite{DueckKorner79}.
Note, that the correct-decoding exponent of the ML decoder (\ref{eqCorrectD0})
is different from the corresponding exponent of the decoder (\ref{eqErrorEvent}), (\ref{eqCorrectExtended}) with $D=0$.
The difference is the result of the tie-breaking, the ML decoder performs.
Without tie-breaking, the decoding can be termed as {\em strict}.
The two nondecreasing curves, given by (\ref{eqCorrectD0}) and (\ref{eqCorrectExtended}), with $D=0$,
both as a function of $R$,
coincide for slopes $< 1$. Then, for greater $R$, the correct decoding exponent of the ML decoder
continues to increase linearly, with constant slope $=1$,
while the exponent of the strict decoder (\ref{eqCorrectExtended}), with $D=0$, eventually becomes $+\infty$.

\section{Random coding error exponent of \\the erasure/list optimum tradeoff decoder} \label{Tradeoff}
In \cite[eq.~11]{Forney68} the decoding error event, given that message $m$ is transmitted, is defined as
\begin{equation} \label{eqForneyErrorEvent}
{\cal E}_{m} \; \triangleq \; \bigg\{\ln\;\frac{P({\bf Y} \,|\, {\bf X\mathstrut}_{m})}{\sum_{m'\,\neq\,m}P({\bf Y} \,|\, {\bf X\mathstrut}_{m'})} \;\; < \;\; nD\bigg\},
\end{equation}
which is different than the error event of the simplified decoder (\ref{eqErrorEvent}).
We derive an exact i.i.d. random coding error exponent for this decoder, for all values of the threshold $D$. The result is given by the minimum of two terms.
One of the terms is given by (\ref{eqExplErr}) and corresponds to the error exponent
of the source/channel duality decoder,
and the other term is very similar, defined as
\begin{equation} \label{eqEDefined}
{E\mathstrut}^{e}(Q, R, D)
\triangleq
\sup_{\rho\,\geq\,0}
\;
\sup_{0 \, \leq \, s \,\leq \,1}
\;
\big\{ {E\mathstrut}_{0}(s, \rho, Q, D) - \rho R \big\}.
\end{equation}
\begin{thm}[Error exponent for optimum tradeoff decoder]\label{thm6}
\begin{align}
& \sup_{Q(x)} \; \lim_{n \, \rightarrow \, \infty} \; \left\{-\frac{1}{n}\ln \Pr \, \{{\cal E}_{m}\}\right\}
\;\; =
\nonumber \\
& \sup_{Q(x)} \; \min \big\{ {E\mathstrut}^{e}(Q, R, D),\;
{E}_{e}(Q, R, D)\big\},
\label{eqExactForney}
\end{align}
{\em for all $(R, D)$, with the possible exception of
points with $R = -D$,
and points with
$0 >D \in{\{D_{\min}(Q)\}\mathstrut}_{Q}\,$
(for $R >  -D$),
where ${\{D_{\min}(Q)\}\mathstrut}_{Q}\,$ is a finite set,
where still}
\begin{align}
& \sup_{Q(x)} \; \liminf_{n \, \rightarrow \, \infty} \; \left\{-\frac{1}{n}\ln \Pr \, \{{\cal E}_{m}\}\right\}
\;\; \geq
\nonumber \\
& \sup_{Q(x)} \; \min \big\{ {E\mathstrut}^{e}(Q, R, D),\;
{E}_{e}(Q, R, D)\big\},
\nonumber \\
& \sup_{Q(x)} \; \limsup_{n \, \rightarrow \, \infty} \; \left\{-\frac{1}{n}\ln \Pr \, \{{\cal E}_{m}\}\right\}
\;\; \leq
\nonumber \\
& \sup_{Q(x)} \; \lim_{\epsilon\,\rightarrow\,0}\;
\min \big\{{E\mathstrut}^{e}(Q, \, R \, - \, \epsilon, \, D),\;
{E}_{e}(Q, \, R, \, D \, - \, \epsilon)\big\},
\nonumber
\end{align}
{\em with ${E\mathstrut}^{e}(Q, R, D)$ and ${E}_{e}(Q, R, D)$ given explicitly by (\ref{eqExplErr}) and (\ref{eqEDefined}).}
\end{thm}

For comparison, the random coding lower bound given by \cite[eq.~24]{Forney68},
can be written,
without maximization over $Q$, in our present terms (with $D$ in place of $T$ and a differently defined $s$, not scaled by $\rho$) as
\begin{equation} \label{eqForneyLowerBound}
{E}_{bound}(Q, R, D) =
\sup_{0 \, \leq \, \rho \,\leq \,1}
\;
\sup_{0 \, \leq \, s \,\leq \,1}
\;
\big\{ {E\mathstrut}_{0}(s, \rho, Q, D) - \rho R \big\}.
\end{equation}
Observe, that indeed
\begin{displaymath}
\min \big\{ {E\mathstrut}^{e}(Q, R, D),\;
{E}_{e}(Q, R, D)\big\}
\; \geq \;
{E}_{bound}(Q, R, D).
\end{displaymath}
The next lemma shows, that the exact exponents (\ref{eqMaxErr}), (\ref{eqExactForney}),
and Forney's lower bound,
given by the maximum of (\ref{eqForneyLowerBound}) over $Q$,
coincide for $D\geq 0$ (erasure regime).
\begin{lemma} \label{lemma1}
For $D \geq 0$
\begin{displaymath}
{E}_{e}(Q, R, D)
\; = \;
{E}_{bound}(Q, R, D)
\end{displaymath}
\end{lemma}
\begin{proof}
\begin{align}
& {E\mathstrut}_{0}\big(s\, = \,\tfrac{1}{1\,+\,\rho}, \,\rho, \,Q, \,D\big)
\;\; =
\nonumber \\
& -\ln\sum_{x, \,y}Q(x)P(y \,|\, x)\left[\sum_{\hat{x}}Q(\hat{x})\left[\frac{P(y \,|\, x)}{P(y \,|\, \hat{x})}\right]^{-\frac{1}{1\,+\,\rho}}\right]^{\rho}
\!\!\!\!
- \tfrac{\rho}{1 \, + \, \rho} D
\nonumber \\
&
\overset{(*)}{\geq} \;
-\ln\sum_{x, \,y}Q(x)P(y \,|\, x)\left[\sum_{\hat{x}}Q(\hat{x})\left[\frac{P(y \,|\, x)}{P(y \,|\, \hat{x})}\right]^{-s}\right]^{\rho}
\!\!\!
- s \rho D
\nonumber \\
& = \;\;
{E\mathstrut}_{0}(s, \rho, Q, D), \;\;\;\;\;\; s \, \geq \, \tfrac{1}{1\,+\,\rho},
\nonumber
\end{align}
where ($*$) holds by H\"older's inequality for $s \, \geq \, \tfrac{1}{1\,+\,\rho}$ and $D\,\geq\,0$.
We conclude, that
\begin{align}
& \sup_{0\,\leq\,\rho\,\leq\,1}
\;\;\;\,
\sup_{s\,\geq\,0}\;
\;\;
\big\{{E\mathstrut}_{0}(s, \rho, Q, D) \; - \; \rho R \big\}
\; =
\nonumber \\
&
\sup_{0\,\leq\,\rho\,\leq\,1}\;
\sup_{0\,\leq\,s\,\leq\,1}\;
\big\{{E\mathstrut}_{0}(s, \rho, Q, D) \; - \; \rho R \big\}.
\nonumber
\end{align}
\end{proof}
As for the $D < 0$ case (list decoding regime),
we note that
the exponent (\ref{eqMaxErr}) becomes $+\infty$ for $D < 0$, and
the exponent (\ref{eqExactForney}) becomes $+\infty$ for $0 < R < -D$,
while Forney's lower bound stays finite. For details see \cite{TridenskiZamir17}.

\section{A selected proof} \label{Derivation}
Here we derive a lower bound on the encoding failure exponent,
which, together with an upper bound, derived in \cite{TridenskiZamir17},
results in Theorem~\ref{thm4}.
Due to the lack of space, all the other proofs are deferred to \cite{TridenskiZamir17}.

The derivation uses a generic auxiliary lemma:
\begin{lemma} \label{lemma2}
{\em Let $Z_{m}\,\sim\, \text{i.i.d}\;\text{Bernoulli}\left({e\mathstrut}^{-nI}\right)$,
$m \, = \, 1, \, 2, \, ... \, , \, {e\mathstrut}^{nR}$. If $I\,\leq\,R \, - \, \epsilon$, with $\epsilon \, > \, 0$, then}
\begin{equation} \label{eqBoundForZero}
\Pr\,\Bigg\{\sum_{m \, = \, 1}^{{e\mathstrut}^{nR}}Z_{m} \, = \, 0\Bigg\}
\;\; < \;\; \exp\big\{-{e\mathstrut}^{n\epsilon}\big\}.
\end{equation}
\end{lemma}

We use the method fo types \cite{CsiszarKorner}, with notation ${P\mathstrut}_{\bf x}^{}$, $\,T({P\mathstrut}_{\bf x}^{})$
for types and type classes, and
${P\mathstrut}_{\hat{\bf x} \, | \,{\bf x}}^{}$,
$\,T({P\mathstrut}_{\hat{\bf x} \, | \,{\bf x}}^{}, \, {\bf X})$
for conditional types and type classes, respectively.
We upper-bound the probability of encoding failure as follows:
\begin{align}
&
\Pr\{{\cal E}_{f}\}
\; \leq
\sum_{{P\mathstrut}_{\bf x}^{}:\;\;
{R\mathstrut}^{\,\text{types}}({P\mathstrut}_{\bf x}^{}, \, Q, \, D)\; \leq \; R \, - \, 2\epsilon_{1}
}
\,\underbrace{\Pr\,\big\{{\bf X} \, \in \, T({P\mathstrut}_{\bf x}^{})\big\}}_{\leq\, 1}
\,\times
\nonumber
\end{align}
\begin{align}
&
\min_{\substack{{P\mathstrut}_{\hat{\bf x} \, | \,{\bf x}}^{}:\\
\\
d({P\mathstrut}_{{\bf x},\,\hat{\bf x}}^{}) \; \leq \; D
}}
\!\!\Pr\,\bigg\{
\sum_{m}
\mathbbm{1}_{\big\{\hat{\bf X\mathstrut}_{m} \; \in \; T({P\mathstrut}_{\hat{\bf x} \, | \,{\bf x}}^{}, \, {\bf X})\big\}}(m)
= 0
\,\bigg|\, {P\mathstrut}_{\bf x}^{}\bigg\}
\nonumber \\
&
\;\;\;\;\;\;\;\;\;\;\;\;\;
+ \sum_{{P\mathstrut}_{\bf x}^{}:\;\;
{R\mathstrut}^{\,\text{types}}({P\mathstrut}_{\bf x}^{}, \, Q, \, D)\; \geq \; R \, - \, 2\epsilon_{1}}
\,\Pr\,\big\{{\bf X} \, \in \, T({P\mathstrut}_{\bf x}^{})\big\}
\nonumber \\
& \;\;\;\;\;\,  = \; S_{1} \, + \, S_{2}.
\label{eqS1S2}
\end{align}
\begin{align}
& S_{1} \overset{(a)}{\leq}
\sum_{\substack{{P\mathstrut}_{\bf x}^{}:\\
\\
{R\mathstrut}^{\,\text{types}}({P\mathstrut}_{\bf x}^{}, \, Q, \, D)\; \leq \; R \, - \, 2\epsilon_{1}
}}\;
\min_{\substack{{P\mathstrut}_{\hat{\bf x} \, | \,{\bf x}}^{}:\\
\\
d({P\mathstrut}_{{\bf x},\,\hat{\bf x}}^{}) \; \leq \; D
}}
\!\!\Pr\Bigg\{\sum_{m \, = \, 1}^{{e\mathstrut}^{nR}}\!Z_{m} = 0\Bigg\}
\nonumber \\
& \overset{(b)}{\leq}
\;
\sum_{{P\mathstrut}_{\bf x}^{}:\;\;
{R\mathstrut}^{\,\text{types}}({P\mathstrut}_{\bf x}^{}, \, Q, \, D)\; \leq \; R \, - \, 2\epsilon_{1}
}
\;\;\;\;\;\;\;\;\;
\;\;\;\;\,
\Pr\Bigg\{\sum_{m \, = \, 1}^{{e\mathstrut}^{nR}} \! B_{m} = 0\Bigg\}
\nonumber \\
& \overset{(c)}{\leq} \;
\sum_{{P\mathstrut}_{\bf x}^{}}
\exp\big\{-{e\mathstrut}^{n\epsilon_{1}}\big\}
\;\leq \;
{(n\, + \, 1)\mathstrut}^{|{\cal X}|}
\cdot
\exp\big\{-{e\mathstrut}^{n\epsilon_{1}}\big\}.\!\!
\label{eqPfUpperBoundS1}
\end{align}
\begin{align}
& S_{2}
\overset{(d)}{\leq}
\;\,
\sum_{{P\mathstrut}_{\bf x}^{}:\;\;
{R\mathstrut}^{\,\text{types}}({P\mathstrut}_{\bf x}^{}, \, Q, \, D)\; \geq \; R \, - \, 2\epsilon_{1}}
\;\;\;\,
\exp\big\{-nD({P\mathstrut}_{\bf x}^{}\;\|\;P )\big\}
\nonumber \\
& \overset{(e)}{\leq}
\;\;
\sum_{{P\mathstrut}_{\bf x}^{}:\;\;
R({P\mathstrut}_{\bf x}^{}, \, Q, \, D \, - \, \epsilon_{2})\; \geq \; R \, - \, 2\epsilon_{1} \, - \, \epsilon_{2}}
\exp\big\{-nD({P\mathstrut}_{\bf x}^{}\;\|\;P )\big\}
\nonumber \\
& \overset{(f)}{\leq} \;
{(n\, + \, 1)\mathstrut}^{|{\cal X}|}
\exp\big\{-n E_{f}^{\,\text{types}}(R - 2\epsilon_{1} - \epsilon_{2}, \, D - \epsilon_{2})\big\}
\nonumber \\
& \overset{(g)}{\leq} \;
{(n\, + \, 1)\mathstrut}^{|{\cal X}|}
\exp\big\{-n E_{f}(R - 2\epsilon_{1} - \epsilon_{2}, \, D - \epsilon_{2})\big\}.\!\!
\label{eqPfUpperBound}
\end{align}
Explanation of steps:
($a$)   holds for sufficiently large $n$, when
\begin{align}
& \Pr\,\big\{\hat{\bf X\mathstrut}_{m} \; \in \; T\big({P\mathstrut}_{\hat{\bf x} \, | \,{\bf x}}^{}, \, {\bf X}\big)
\;\big|\;
{\bf X} \, \in \, T({P\mathstrut}_{\bf x}^{})\big\}
\; \geq
\nonumber \\
&
\exp\Big\{-n\left[D\big({P\mathstrut}_{{\bf x}, \, \hat{\bf x}}^{} \,\big\|\, {P\mathstrut}_{\bf x}^{} \!\times Q\big) \, + \, \epsilon_{1}\right]\Big\},
\nonumber
\end{align}
with
\begin{displaymath}
Z_{m} \sim 
\text{Bernoulli}\left(
\exp\Big\{-n\left[D\big({P\mathstrut}_{{\bf x}, \, \hat{\bf x}}^{} \,\big\|\, {P\mathstrut}_{\bf x}^{} \!\times Q\big) + \epsilon_{1}\right]\Big\}
\right).
\end{displaymath}
($b$) holds for
\begin{displaymath}
B_{m} \sim
\text{Bernoulli}\left(
\exp\Big\{-n\big[{R\mathstrut}^{\,\text{types}}({P\mathstrut}_{\bf x}^{}, Q, D) \, + \, \epsilon_{1}\big]\Big\}
\right),
\end{displaymath}
where
\begin{equation} \label{eqRTypes}
{R\mathstrut}^{\,\text{types}}({P\mathstrut}_{\bf x}^{}, Q, D)
\; \triangleq \;
\min_{{P\mathstrut}_{\hat{\bf x} \, | \,{\bf x}}^{}: \;\;
d({P\mathstrut}_{{\bf x},\,\hat{\bf x}}^{}) \; \leq \; D
}\;
D\big({P\mathstrut}_{{\bf x}, \, \hat{\bf x}}^{} \,\big\|\, {P\mathstrut}_{\bf x}^{} \!\times Q\big).
\end{equation}
($c$) holds by Lemma~\ref{lemma2} for
\begin{displaymath}
I \;\; = \;\; {R\mathstrut}^{\,\text{types}}({P\mathstrut}_{\bf x}^{}, Q, D) \, + \, \epsilon_{1}
\;\; \leq \;\; R \, - \, 2\epsilon_{1} \, + \, \epsilon_{1}
\;\; = \;\; R \, - \, \epsilon_{1}.
\end{displaymath}
($d$) uses the upper bound on the probability of a type.\newline
($e$) Let ${W\mathstrut}^{*}$ denote the conditional distribution, achieving $R({P\mathstrut}_{\bf x}^{}, \, Q, \, D \, - \, \epsilon_{2})<+\infty$
for some $\epsilon_{2}\,>\,0$.
This implies
\begin{align}
D({P\mathstrut}_{\bf x}^{} \circ {W\mathstrut}^{*} \,\|\, {P\mathstrut}_{\bf x}^{} \!\times Q)
\;\; & = \;\; R({P\mathstrut}_{\bf x}^{}, \, Q, \, D \, - \, \epsilon_{2}),
\label{eqComb1} \\
d({P\mathstrut}_{\bf x}^{} \circ {W\mathstrut}^{*}) \;\; & \leq \;\; D \, - \, \epsilon_{2}.
\nonumber
\end{align}
Let ${W\mathstrut}_{n}^{*}$
denote a quantized version of the conditional distribution ${W\mathstrut}^{*}$
with variable precision $1/\big(n {P\mathstrut}_{\bf x}^{}(x)\big)$, i.e. a set of types with denominators $n {P\mathstrut}_{\bf x}^{}(x)$,
such that the joint distribution ${P\mathstrut}_{\bf x}^{} \circ {W\mathstrut}_{n}^{*}$ is a type with denominator $n$.
Observe, that the differences between ${P\mathstrut}_{\bf x}^{} \circ {W\mathstrut}^{*}$ and ${P\mathstrut}_{\bf x}^{} \circ {W\mathstrut}_{n}^{*}$
do not exceed $\tfrac{1}{n}$.
Therefore,
since the divergence, as a function of ${P\mathstrut}_{\bf x}^{} \circ W$, has bounded derivatives, and also the distortion measure $d(x, \hat{x})$
is bounded, for any $\epsilon_{2}\,>\,0$ there exists $n$ large enough,
such that the quantized distribution ${W\mathstrut}_{n}^{*}$
satisfies
\begin{align}
D({P\mathstrut}_{\bf x}^{} \circ {W\mathstrut}_{n}^{*} \,\|\, {P\mathstrut}_{\bf x}^{} \!\times Q)
\;\; & \leq \;\; D({P\mathstrut}_{\bf x}^{} \circ {W\mathstrut}^{*} \,\|\, {P\mathstrut}_{\bf x}^{} \!\times Q) \, + \, \epsilon_{2},
\label{eqComb2} \\
d({P\mathstrut}_{\bf x}^{} \circ {W\mathstrut}_{n}^{*}) \;\; & \leq \;\; D.
\nonumber
\end{align}
The last inequality implies
\begin{align}
D({P\mathstrut}_{\bf x}^{} \circ {W\mathstrut}_{n}^{*} \,\|\, {P\mathstrut}_{\bf x}^{} \!\times Q)
\;\; & \geq \;\; {R\mathstrut}^{\,\text{types}}({P\mathstrut}_{\bf x}^{}, Q, D).
\label{eqComb3}
\end{align}
The relations (\ref{eqComb3}), (\ref{eqComb2}), (\ref{eqComb1}) together give
\begin{align}
R({P\mathstrut}_{\bf x}^{}, \, Q, \, D \, - \, \epsilon_{2}) \, + \, \epsilon_{2}
\;\; & \geq \;\; {R\mathstrut}^{\,\text{types}}({P\mathstrut}_{\bf x}^{}, Q, D).
\label{eqResult4}
\end{align}
This explains ($e$).
($f$) uses the definition
\begin{equation} \label{eqEfTypes}
E_{f}^{\,\text{types}}(R, D)
\;\; \triangleq \;\;
\min_{{P\mathstrut}_{\bf x}^{}: \;\;
R({P\mathstrut}_{\bf x}^{}, \,Q, \,D) \; \geq \; R} \; D({P\mathstrut}_{\bf x}^{} \;\|\; P).
\end{equation}
($g$) $E_{f}^{\,\text{types}}(R, D)$ is bounded from below by $E_{f}(R, D)$ defined in (\ref{eqEfDefinition}).
We conclude from (\ref{eqS1S2}), (\ref{eqPfUpperBoundS1}), (\ref{eqPfUpperBound}):
\begin{align}
& \liminf_{n \, \rightarrow \, \infty} \; \left\{-\frac{1}{n}\ln \Pr\{{\cal E}_{f}\}\right\}
\;\; \geq \;\;
\lim_{\epsilon\,\rightarrow\,0}\;
E_{f}(R \, - \, \epsilon, \, D \, - \, \epsilon).
\nonumber
\end{align}


\bibliographystyle{IEEEtran}
\bibliography{expduality}

\begin{thebibliography}{10}
\providecommand{\url}[1]{#1}
\csname url@samestyle\endcsname
\providecommand{\newblock}{\relax}
\providecommand{\bibinfo}[2]{#2}
\providecommand{\BIBentrySTDinterwordspacing}{\spaceskip=0pt\relax}
\providecommand{\BIBentryALTinterwordstretchfactor}{4}
\providecommand{\BIBentryALTinterwordspacing}{\spaceskip=\fontdimen2\font plus
\BIBentryALTinterwordstretchfactor\fontdimen3\font minus
  \fontdimen4\font\relax}
\providecommand{\BIBforeignlanguage}[2]{{%
\expandafter\ifx\csname l@#1\endcsname\relax
\typeout{** WARNING: IEEEtran.bst: No hyphenation pattern has been}%
\typeout{** loaded for the language `#1'. Using the pattern for}%
\typeout{** the default language instead.}%
\else
\language=\csname l@#1\endcsname
\fi
#2}}
\providecommand{\BIBdecl}{\relax}
\BIBdecl

\bibitem{GuptaVerdu11}
A.~Gupta and S.~Verd\'u, ``Operational duality between lossy compression and
  channel coding,'' \emph{IEEE Trans. on Information Theory}, vol.~57, no.~6,
  pp. 3171--3179, Jun. 2011.

\bibitem{PradhanChouRamchandran03}
S.~Pradhan, J.~Chou, and K.~Ramchandran, ``Duality between source coding and
  channel coding and its extension to the side information case,'' \emph{IEEE
  Trans. on Information Theory}, vol.~49, no.~5, pp. 1181--1203, May 2003.

\bibitem{Blahut74}
R.~Blahut, ``Hypothesis testing and {I}nformation {T}heory,'' \emph{IEEE Trans.
  on Information Theory}, vol.~20, no.~4, pp. 405--417, Jul. 1974.

\bibitem{JanaBlahutISIT2006}
S.~Jana and R.~Blahut, ``Insight into source/channel duality and more based on
  an intuitive covering/packing lemma,'' in \emph{IEEE International Symposium
  on Information Theory}, Jul. 9-14 2006, pp. 2329--2333.

\bibitem{ZamirRose01}
R.~Zamir and K.~Rose, ``Natural type selection in adaptive lossy compression,''
  \emph{IEEE Trans. on Information Theory}, vol.~47, no.~1, pp. 99--111, Jan.
  2001.

\bibitem{Marton74}
K.~Marton, ``Error exponent for source coding with a fidelity criterion,''
  \emph{IEEE Trans. on Information Theory}, vol.~20, no.~2, pp. 197--199, Mar.
  1974.

\bibitem{CsiszarKorner}
I.~Csisz\'ar and J.~K{\"o}rner, \emph{Information Theory: Coding Theorems for
  Discrete Memoryless Systems}.\hskip 1em plus 0.5em minus 0.4em\relax Academic
  Press, 1981.

\bibitem{TridenskiZamirITW2015}
S.~Tridenski and R.~Zamir, ``Stochastic interpretation for the {A}rimoto
  algorithm,'' in \emph{IEEE Information Theory Workshop}, Apr. 2015.

\bibitem{Forney68}
G.~D. Forney, Jr., ``Exponential error bounds for erasure, list, and decision
  feedback schemes,'' \emph{IEEE Trans. on Information Theory}, vol.~14, no.~2,
  pp. 206--220, Mar. 1968.

\bibitem{WeinbergerMerhavISIT2015}
N.~Weinberger and N.~Merhav, ``Simplified erasure/list decoding,'' in
  \emph{IEEE International Symposium on Information Theory}, Jun. 2015, pp.
  2226--2230.

\bibitem{BaruchMerhav11}
A.~S. Baruch and N.~Merhav, ``Exact random coding exponents for erasure
  decoding,'' \emph{IEEE Trans. on Information Theory}, vol.~57, no.~10, pp.
  6444--6454, Oct. 2011.

\bibitem{Merhav08}
N.~Merhav, ``Error exponents of erasure/list decoding revisited via moments of
  distance enumerators,'' \emph{IEEE Trans. on Information Theory}, vol.~54,
  no.~10, pp. 4439--4447, Oct. 2008.

\bibitem{TridenskiZamir17}
S.~Tridenski and R.~Zamir, ``Analogy and duality between random channel coding
  and lossy source coding,'' arxiv.

\bibitem{Gallager73}
R.~G. Gallager, ``The random coding bound is tight for the average code,''
  \emph{IEEE Trans. on Information Theory}, vol.~19, no.~2, pp. 244--246, Mar.
  1973.

\bibitem{DueckKorner79}
G.~Dueck and J.~K{\"o}rner, ``Reliability function of a discrete memoryless
  channel at rates above capacity,'' \emph{IEEE Trans. on Information Theory},
  vol.~25, no.~1, pp. 82--85, Jan. 1979.

\end{thebibliography}

\end{document}